\documentclass[12pt, reqno]{amsart}
\usepackage{amsmath, amsthm, amscd, amsfonts, amssymb, graphicx, setspace,color}
\usepackage{sagetex}
\usepackage[letterpaper,hmargin=1.25in,vmargin=1.5in]{geometry}
\usepackage[bookmarksnumbered, colorlinks, plainpages]{hyperref}
\hypersetup{colorlinks=true,linkcolor=red, anchorcolor=green, citecolor=cyan, urlcolor=red, filecolor=magenta, pdftoolbar=true}
\newtheorem{theorem}{Theorem}[section]

\newtheorem{proposition}[theorem]{Proposition}

\theoremstyle{definition}
\newtheorem{definition}[theorem]{Definition}

\theoremstyle{remark}
\newtheorem{remark}[theorem]{Remark}
\numberwithin{equation}{section}
\begin{document}
\begin{sagesilent}
import numpy as np
def sum(n,r):
    sum=0
    for k in range(1,r+1):
        sum = sum + binomial(r-1,k-1)*(k**(n-1))*((-1)**(r-k))
    return sum   
    
def a(n,r):
    return factorial(r-1)*sum(n,r)  
\end{sagesilent}

\setcounter{page}{1}

\title[Study of Multiplication operator]{\textbf{Study of  Multiplication operator on $\mathcal{H}_{\frac{1}{2}}\oplus\mathcal{H}_{-\frac{1}{2}}$}}
\author[A.Noufal]{\tiny{A. Noufal \\ Department of Mathematics,\\
Cochin University of Science and Technology, Cochin -682 022, India. \\
email: noufalasharaf$@$gmail.com}}
\begin{abstract}
In this paper we focus on the continuous representation
on $\mathcal{S}(\mathbb{R})\subset L^2(\mathbb{R})$ with the operators $\frac{(PQ+QP)}{4}$ and $\frac{Q^2}{2}$ as generators given by $U[p,q]=\exp(-\frac{iqQ^2}{2})
\exp(i\log p\frac{PQ+QP}{4}).$ Action of the operator $\exp(PQ+QP)$ and the unitary equivalent operator on $\mathcal{S}(\mathbb{R})\subseteq L^2(\mathbb{R})$ of the multiplication operator in
$\mathcal{H}_{\frac{1}{2}}\oplus\mathcal{H}_{-\frac{1}{2}}$ is  
obtained.
\end{abstract}
\maketitle
\section{Introduction}
The general concepts and properties of continuous
representation theory have been developed by Klauder\cite{JRK} in 1963.
Basically any quantization procedure is a method  for relating the quantum 
problem to its classical counter part. The ingredients in the basic 
structure of quantum mechanics are the unit vectors in a Hilbert space 
corresponds to the states of the system and automorphisms among unit vectors.
The general postulates of Continuous Representation of an infinite dimensional 
Hilbert space is considered as in \cite{JRK}. As per the formulation in \cite{CRT},
we consider $U[l]$ be a family of unitary operators on $\mathcal{H}.$ 
Choose an arbitrary but fixed vector $\varphi_0\in \mathcal{H}$ called the 
fiducial vector. We can generate a subset of $\mathcal{H}$ by operating these
family of operators $U[l]$ on the fiducial vector. For any vector $\Psi\in\mathcal{H}$,
we can associate the complex, bounded continuous function 
$\psi(l) = \langle \Psi, U[l]\varphi_0 \rangle$ and the set 
$\{\psi(l) :\Psi\in\mathcal{H}\}$ is called the continuous representation
of the Hilbert space $\mathcal{H}.$ It is quite useful if we take $U[l]$ as 
the elements of a kinematic group and to interpret the labels as the classical
canonical coordinates $(p,q)$ for a system with one degree of freedom.
For the classical Cartesian pair $(p,q)\in \mathbb{R}^2,$ CRT has been developed by
Klauder and McKenna\cite{JMK}. In this development they have used the Weyl operators
$U[p,q] = \exp[ipQ -qP]$, where $Q$ as the familiar multiplication operator, $(Q f)(x)=xf(x)$ and $P$ 
the differential operator $(P f)(x)=-if^{\prime}(x)$ which are self adjoint operators satisfying the 
canonical commutation relation $[Q, P] = iI.$ This relation implies that $P$
and $Q$ must have a spectrum on the real line if they are self adjoint, if the spectrum of 
either $P$ or $Q$ is restricted, at least one of the operator looses the self
ad-jointness, say $Q$ is not self adjoint and therefore have no spectral resolution 
and the appropriate unitary operators cannot be the familiar Weyl operators.
This leads to consider a different pair of operators $P$ and $B$, both of which are
self adjoint and obey the commutation relation $[B,P] = iP$ as in \cite{CRT}.
For the coordinates with restriction $(p,q)\in\mathbb{R}^+\times\mathbb{R},$ 
consider the group of all linear transformations without reflections on the real line
known as affine group is taken and the unitary group elements are given as
$U[p,q] = \exp[-iqP]\exp[i\log pB],$ where $P$ and $B$ are self adjoint generators
which satisfy $[B,P] = iP.$ This group is formally "close" to the canonical group, because we obtain 
this group by multiplying the commutation relation $[Q,P] =iI$ by $P$ on either side
and identity $B = \frac{1}{2}(PQ+QP)$ which is self adjoint. So to study the 
continuous representation of $\mathcal{H}$ using this family of operators we want to 
understand the action of $\exp(imP)$ and  $\exp(PQ+QP)$. In this paper we prove 
the action of these operators. The following proposition proves an important identity 
which is useful in the series representation of $\exp(QP).$

\begin{proposition}\label{first}
For the Heisenberg operators $Q$ and $P$ \[QP(Q^nP^n)=Q^{n+1}P^{n+1}+inQ^nP^n,\hspace{0.2cm}\text{for all}\hspace{0.2cm} n\in\mathbb{N}\]
\end{proposition}
Proof :For $n=1,$ $QP(QP)=Q(PQ)P=Q(QP+i)P=Q^2P^2+iQP$, assume the result is true for $n=k,$ $QP(Q^kP^k)=Q^{k+1}P^{k+1}+ikQ^kP^k$
\begin{align*} 
QP(Q^{k+1}P^{k+1})&=QP(QQ^{k}P^{k}P)=Q(PQ)Q^kP^kP=Q(QP+i)Q^kP^kP\\
&=Q[Q^{k+1}P^{k+1}+ikQ^kP^k]P+iQ^{k+1}P^{k+1}\\
&=Q^{k+2}P^{k+2}+i(k+1)Q^{k+1}P^{k+1}
\end{align*}
by mathematical induction $QP(Q^nP^n)=Q^{n+1}P^{n+1}+inQ^nP^n$ 
for $n\in\mathbb{N}$
When we focus on the action of the operator
$\exp(QP) = I+\frac{QP}{1!}+\frac{(QP)^2}{2!}+\frac{(QP)^3}{3!}+\cdots$ 
we need to handle the operators of the form $(QP)^n$ for $n\in\mathbb{N}$,
the following identity gives an easy procedure to convert those operators as the sum
of operators of the type $Q^nP^n$ and a method of computing the coefficients is also given. Proposition \ref{first}
will help us to compute the higher powers of $QP.$ For simplicity we consider the differentiation operator
$D,$ and we replace  $P = -iD.$
\begin{proposition}\label{imprp}
For $n\in\mathbb{N},$ the operators $Df(x)=f^{\prime}(x)$ and $Qf(x)=xf(x),$ 
\begin{equation*}(QD)^n=\sum_{r=1}^{\infty}a_{n,r}Q^rD^r,\hspace{0.2cm}where\hspace{0.2cm} a_{n,r}=\frac{1}{(r-1)!}\sum_{k=1}^{r}\binom{r-1}{k-1}(-1)^{r-k}k^{n-1}
\end{equation*}
\end{proposition}
Proof: 
 For $n=1,$\[a_{1,r}=\frac{1}{(r-1)!}\sum_{k=1}^{r}\binom{r-1}{k-1}(-1)^{r-k}\]
$a_{1,1}=\frac{1}{0!}\cdot 1=1$ and for $r>1$\[a_{1,r}=\frac{1}{(r-1)!}\sum_{k=1}^{r}\binom{r-1}{k-1}(-1)^{r-k}=
\frac{1}{(r-1)!}(1-1)^{r-1}=0\]
Therefore \[\sum_{r=1}^{\infty}a_{1,r}Q^rD^r=a_{1,1}QD=QD\]
Assume the result is true for $n=m$ i.e., $\displaystyle(QD)^m=\sum_{r=1}^{\infty}a_{m,r}Q^rD^r.$
\begin{align*}
(QD)^{m+1}&=QD(QD)^m=QD\left[\sum_{r=1}^{\infty}a_{m,r}Q^rD^r\right]\\
&=\sum_{r=1}^{\infty}a_{m,r}QD(Q^rD^r)=\sum_{r=1}^{\infty}a_{m,r}\left[Q^{r+1}D^{r+1}+rQ^rD^r\right]\\
&=\sum_{r=1}^{\infty}a_{m,r}Q^{r+1}D^{r+1}+r\sum_{r=1}^{\infty}a_{m,r}Q^rD^r\\
&=\sum_{r=2}^{\infty}a_{m,r-1}Q^{r}D^{r}+r\sum_{r=2}^{\infty}a_{m,r}Q^rD^r+a_{m,1}QD\\
&=a_{m,1}QD+\sum_{r=2}^{\infty}\left(a_{m,r-1}+ra_{m,r}\right)Q^rD^r
\end{align*}
Since $a_{n,1}=1$ for all $n,$ in particular $a_{m,1}=a_{m+1,1}=1.$ For $r=2,3,\cdots$
\begin{equation}\label{recurse}
 a_{m+1,r}=a_{m,r-1}+ra_{m,r}\end{equation}
For n = 2,  $\displaystyle(QD)^2=\sum_{r=1}^{\infty}a_{2,r}Q^rD^r$, where
$a_{2,1} =\sage{a(2,1).n(digits=1)},\, a_{2,2}=\sage{a(2,2).n(digits=1)},\, a_{2,3} = \sage{a(2,3).n(digits=1)},\, a_{2,4}=\sage{a(2,4).n(digits=1)},$
etc.
For n = 3, $\displaystyle(QD)^3=\sum_{r=1}^{\infty}a_{3,r}Q^rD^r$, where
$a_{3,1} =\sage{a(3,1).n(digits=1)},\, a_{3,2}=\sage{a(3,2).n(digits=1)},\, a_{3,3} = \sage{a(3,3).n(digits=1)},\, a_{3,4}=\sage{a(3,4).n(digits=1)},$
etc. Recursive relation (\ref{recurse}) can be used to find the coefficients of $(QD)^n$ using the coefficients
of $(QD)^{n-1}$ for any $n\in\mathbb{N}.$
For any $f\in\mathcal{S}(\mathbb{R})\subset L^2(\mathbb{R})$, action of the operators $\exp(mQD),$ $\exp(imP)$,
$\exp(imQ^2)$ 
will be $f\mapsto f(e^mt),$ $f\mapsto f(x+m),$ $f\mapsto \exp(imx^2)f(x)$
respectively. All but the first one are easy observations, for proving the first one we use 
the proposition (\ref{imprp}). Look at the series expansion of $\exp(mQD)=I+\frac{mQD}{1!}+\frac{(mQD)^2}{2!}+\cdots$
It's action on an element of $\mathcal{S}(\mathbb{R})$ is 
{\small
\begin{align*}
[e^{mQD}f](t)&=f(x)+\frac{m}{1!}QDf(x)+\frac{m^2}{2!}(QD)^2f(x)+\frac{m^3}{3!}(QD)^3f(x)+\cdots \\
&=f(x)+\frac{m}{1!}QDf(x)+\frac{m^2}{2!}\left[\sum_{r=1}^{\infty}a_{2,r}Q^rD^r\right]f(x)\\
&\hspace{1cm}+\frac{m^3}{3!}\left[\sum_{r=1}^{\infty}a_{3,r}Q^rD^r\right]f(x)+\cdots\\
&=f(x)+\frac{m}{1!}QDf(x)+\frac{m^2}{2!}\left[\sum_{r=1}^{\infty}\frac{1}{(r-1)!}\sum_{k=1}^{r}\binom{r-1}{k-1}(-1)^{r-k}kQ^rD^r\right]\\
&\hspace{1cm}f(x)+\frac{m^3}{3!}\left[\sum_{r=1}^{\infty}\frac{1}{(r-1)!}\sum_{k=1}^{r}\binom{r-1}{k-1}(-1)^{r-k}k^2Q^rD^r\right]f(x)+\cdots\\
&=f(x)+[\frac{m}{1!}+\frac{m^2}{2!}+\frac{m^3}{3!}+\cdots+\frac{m^n}{n!}a_{n,1}+\cdots]QDf(x)+[\frac{m^2}{2!}+\frac{3m^3}{3!}\\
&\hspace{1cm}+\frac{7m^4}{4!}+\cdots+\frac{m^n}{n!}a_{n,2}+\cdots]Q^2D^2f(x)+[\frac{m^3}{3!}+\frac{6m^4}{4!}+\frac{25m^5}{5!}\\
&\hspace{1cm}+\cdots+\frac{m^n}{n!}a_{n,3}+\cdots]Q^3D^3f(x)+[\frac{m^4}{4!}+\frac{10m^5}{5!}+\frac{146m^6}{6!}+\cdots
\\&\hspace{1cm}+\frac{m^n}{n!}a_{n,4}+\cdots]Q^4D^4f(x)+\cdots+[\frac{m^r}{r!}+\cdots+\frac{m^n}{n!}a_{n,r}+\cdots]\\
&\hspace{1.2cm}Q^rD^rf(x)+\cdots\\
&=f(x)+\left[\sum_{n=1}^{\infty}\frac{m^n}{n!}a_{n,1}\right]QDf(x)+\left[\sum_{n=1}^{\infty}\frac{m^n}{n!}a_{n,2}\right]Q^2D^2f(x)\\
&\hspace{1cm}+\left[\sum_{n=1}^{\infty}\frac{m^n}{n!}a_{n,3}\right]Q^3D^3f(x)+\left[\sum_{n=1}^{\infty}\frac{m^n}{n!}a_{n,4}\right]Q^4D^4f(x)+\cdots\\
&\hspace{1cm}+\left[\sum_{n=1}^{\infty}\frac{m^n}{n!}a_{n,r}\right]Q^rD^rf(x)+\cdots\\
&=f(x)+\left[\sum_{n=1}^{\infty}\frac{m^n}{n!}\right]QDf(x)+\left[\sum_{n=1}^{\infty}\frac{m^n}{n!}
\left(2^{n-1}-1\right)\right]Q^2D^2f(x)\\
&\hspace{1cm}+\left[\sum_{n=1}^{\infty}\frac{m^n}{n!}\frac{1}{2!}\left(3^{n-1}-\binom{2}{1}2^{n-1}+1\right)\right]Q^3D^3f(x)
\end{align*}
\begin{align*}
&\hspace{1cm}+\left[\sum_{n=1}^{\infty}\frac{m^n}{n!}\frac{1}{3!}\left(4^{n-1}-\binom{3}{1}3^{n-1}+\binom{3}{2}2^{n-1}-1\right)\right]Q^4D^4f(x)+\cdots\\
&\hspace{1cm}+\sum_{n=1}^{\infty}\frac{m^n}{n!}\frac{1}{(r-1)!}[r^{n-1}-\binom{r-1}{1}(r-1)^{n-1}+\binom{r-1}{2}(r-2)^{n-1}+\cdots\\
&\hspace{1.3cm}+(-1)^{r-2}\binom{r-1}{r-2}2^{n-1}+1]Q^rD^rf(x)+\cdots\\
&=f(x)+\frac{(e^m-1)}{1!}QDf(x)+\frac{1}{2!}\left[\sum_{n=1}^{\infty}\left(\frac{(2m)^n}{n!}-\frac{2m^n}{n!}\right)\right]Q^2D^2f(x)+\frac{1}{2!}\\
&\hspace{1cm}\left[\sum_{n=1}^{\infty}\left(\frac{1}{3}\frac{(3m)^n}{n!}-\frac{(2m)^n}{n!}+\frac{m^n}{n!}\right)\right]Q^3D^3f(x)+\frac{1}{3!}\\
&\hspace{1cm}\left[\sum_{n=1}^{\infty}\left(\frac{1}{4}\frac{(4m)^n}{n!}-\frac{(3m)^n}{n!}+\frac{3}{2}\frac{(2m)^n}{n!}-\frac{m^n}{n!}\right)\right]Q^4D^4f(x)+\cdots\\
&\hspace{1cm}+\frac{1}{(r-1)!}\left[\sum_{n=1}^{\infty}\left(\frac{1}{r}\frac{(rm)^n}{n!}-\frac{(rm-m)^n}{n!}+\cdots+(-1)^{r-1}
\frac{m^n}{n!}\right)\right]Q^rD^rf(x)\\&\hspace{1cm}+\cdots\\
&=f(x)+\frac{1}{1!}\left[(e^m-1)\right]QDf(x)+\frac{1}{2!}\left[(e^{2m}-1)-(e^m-1)\right]Q^2D^2f(x)\\
&\hspace{1cm}+\frac{1}{3!}\left[(e^{3m}-1)-3(e^{2m}-1)+3(e^m-1)\right]Q^3D^3f(x)\\
&\hspace{1cm}+\frac{1}{4!}\left[(e^{4m}-1)-4(e^{3m}-1)+6(e^{2m}-1)-4(e^m-1)\right]Q^4D^4f(x)+\cdots\\
&=f(x)+\frac{(e^m-1)}{1!}xf^{\prime}(x)+\frac{(e^m-1)^2}{2!}x^2f^{\prime\prime}(x)+\frac{(e^m-1)^3}{3!}x^2f^{\prime\prime\prime}(x)+\cdots\\
&\hspace{1cm}+\frac{(e^m-1)^r}{r!}x^rf^{(r)}(x)+\cdots\\
&=f(x+(e^m-1)x)=f(e^mx)
\end{align*}
}
Note that, actions of $\left[\exp(imP)\cdot \exp(inQ)\right]$ and $\left[\exp(imQ)\cdot \exp(inP)\right]$ are not the same.
First one is $f\mapsto e^{in(x+m)}f(x+m)$ but the second one is $f\mapsto e^{imx}f(x+n).$
We consider the family of unitary operators \begin{align}U[p,q]:&=\exp\left(-\frac{i}{2}q \, Q^{2} \right)\exp\left(\frac{i}{4}                                        
\log p\,[P Q +Q P]\right)\\&=\exp\left(-\frac{i}{2}q \, Q^{2} \right)\exp\left(\frac{i}{4}\log p (2Q P-i)\right)\end{align}
members in this family are actually 
\begin{align*}\exp\left(-\frac{i}{2}qQ^{2}\right)\exp\left(\frac{\log p}{2} \, Q D+\frac{\log p}{4}\right)
=p^{\frac{1}{4}}\exp\left(-\frac{i}{2}qQ^{2}\right)\exp\left(\frac{ \log p}{2} \, Q D\right)\end{align*}
and the action of each member of the family will be $f\mapsto p^{\frac{1}{4}}\exp\left(-\frac{i}{2}q \, x^{2} \right)f(\sqrt{p}x)$ for any $f\in\mathcal{S}(\mathbb{R}).$
\section{ Space of Analytic functions $\mathcal{H}_{\frac{1}{2}}\oplus\mathcal{H}_{-\frac{1}{2}}$}
Consider $\mathcal{H}_{\frac{1}{2}}$ as the space of analytic functions defined on the open upper half-plane and square integrable with respect to the measure $d\mu_\frac{1}{2}(z)=\left(\frac{z-\overline{z}}{2i}\right)^\frac{1}{2}\frac{dz\wedge d\overline{z}}{2i}$, and $\mathcal{H}_{-\frac{1}{2}}$ as the space of analytic functions defined on the open upper half-plane and square integrable with respect to the measure $d\mu_{-\frac{1}{2}}(z)=\left(\frac{z-\overline{z}}{2i}\right)^{-\frac{1}{2}}\frac{dz\wedge d\overline{z}}{2i}.$ Note that $\mathcal{H}_{\frac{1}{2}}$ and $\mathcal{H}_{-\frac{1}{2}}$
are Hilbert spaces with the innerproducts $\langle f, g\rangle =  \int f(z) \overline{g(z)}\, d\mu_\frac{1}{2}(z)$ and $\langle f, g\rangle =  \int f(z) \overline{g(z)}\, d\mu_{-\frac{1}{2}}(z)$ respectively. By Riesz representation theorem there exist a reprodcing kernel for both $\mathcal{H}_{\frac{1}{2}}$ and $\mathcal{H}_{-\frac{1}{2}}$ i.e., for each elements $z\in\mathbb{C},$ functions $\rho_{\frac{1}{2}}^z$ and $\rho_{-\frac{1}{2}}^z$ exists in $\mathcal{H}_{\frac{1}{2}}$
and $\mathcal{H}_{-\frac{1}{2}}$ respectively with the properties \begin{align}
\int \rho_{\frac{1}{2}}^z(z^\prime) f(z^\prime)d\mu_\frac{1}{2}(z) &= f(z^\prime)                                       \\
\int \rho_{-\frac{1}{2}}^z(z^\prime) f(z^\prime)d\mu_{-\frac{1}{2}}(z) &= f(z^\prime)                                       
\end{align}
note that the space $\mathcal{H}_{\frac{1}{2}}$ is unitary equivalent to the space $D_{\frac{1}{2}}$ of analytic functions on the open unit disk which are square integrable with respect to the measre $d\nu_\frac{1}{2}(\omega) = \left(\frac{1-\omega\overline{\omega}}{2} \right)^{\frac{1}{2}}\, dxdy,$ also the space $\mathcal{H}_{-\frac{1}{2}}$ is unitary equivalent to the space $D_{-\frac{1}{2}}$ of analytic functions on the open unit disk which are square integrable with respect to the measre $d\nu_{-\frac{1}{2}}(\omega) = \left(\frac{1-\omega\overline{\omega}}{2} \right)^{-\frac{1}{2}}\, dxdy.$ Consider the unitary transforms $B^{\frac{1}{2}} : \mathcal{H}_{\frac{1}{2}}\rightarrow D_{\frac{1}{2}}$ and  $B^{-\frac{1}{2}} : \mathcal{H}_{-\frac{1}{2}}\rightarrow D_{-\frac{1}{2}}$ defined by $(B^{\frac{1}{2}}f)(\omega) = 2^{\frac{5}{4}}\left(\frac{1-\omega}{i}\right)^{-\frac{5}{2}}f\left(i\frac{1+\omega}{1-\omega}\right),$ $(B^{-\frac{1}{2}}f)(\omega) = 2^{\frac{3}{4}}\left(\frac{1-\omega}{i}\right)^{-\frac{5}{2}}f\left(i\frac{1+\omega}{1-\omega}\right),$ the system of functions $u_n(\omega) = \dfrac{2^\frac{1}{4}}{\sqrt{\pi}}\sqrt{\frac{\Gamma (n+5/2)}{\Gamma(n+1)\Gamma (3/2)}} \omega^n$  forms a complete orthogonal system for $D_{\frac{1}{2}}$ and the functions $v_n(\omega) = \dfrac{2^{-\frac{1}{4}}}{\sqrt{\pi}}\sqrt{\frac{\Gamma (n+3/2)}{\Gamma(n+1)\Gamma (1/2)}} \omega^n$  forms a complete orthogonal system for $D_{-\frac{1}{2}}.$ According to Bergmann,
\[k_{\frac{1}{2}}^\omega(\omega^\prime)= \sum_{n=0}^{\infty} u_n(\omega)\overline{u_n(\omega^\prime) }= \sum_{n=0}^{\infty}\dfrac{2^\frac{1}{2}}{\pi}\frac{\Gamma (n+5/2)}{\Gamma(n+1)\Gamma (3/2)} ({\omega\bar{\omega}^\prime})^n \]
will be the reproducing kernel for the space $D_{\frac{1}{2}},$ and 
\[k_{-\frac{1}{2}}^\omega(\omega^\prime)= \sum_{n=0}^{\infty} v_n(\omega)\overline{v_n(\omega^\prime) }= \sum_{n=0}^{\infty}\dfrac{2^{-\frac{1}{2}}}{\pi}\frac{\Gamma (n+3/2)}{\Gamma(n+1)\Gamma (1/2)} ({\omega\bar{\omega}^\prime})^n \]
will be the reproducing kernel for the space $D_{-\frac{1}{2}}, $ i.e., for any $g \in D_{\frac{1}{2}},$ we have
\begin{equation}
  \int k_{\frac{1}{2}}^\omega(\omega^\prime)  g(\omega^\prime)\,d\nu_\frac{1}{2}(\omega^\prime) = g(\omega)
\end{equation}
and for the element $g\in  D_{-\frac{1}{2}} $ we have $\int  k_{-\frac{1}{2}}^\omega(\omega^\prime) g(\omega^\prime)\,d\nu_\frac{1}{2}(\omega^\prime) = g(\omega)
.$ It  is  to be noted that \[k_{\frac{1}{2}}^\omega(\omega^\prime) = \frac{3}{\sqrt{2}\pi} \,\,_2 F_1\left(\frac{5}{2}, 1;1;\omega\bar{\omega}^\prime\right)= \frac{3}{8\pi}\left( \frac{1-\omega\bar{\omega}^\prime}{2}\right)^{-\frac{5}{2}}\] and 
\[k_{-\frac{1}{2}}^\omega(\omega^\prime)=\frac{1}{2\pi\sqrt{2}} \,\,_2 F_1\left(\frac{5}{2}, 1;1;\omega\bar{\omega}^\prime\right) = \frac{1}{8\pi}\left( \frac{1-\omega\bar{\omega}^\prime}{2}\right)^{-\frac{3}{2}}\]
by unitary transform $B^{\frac{1}{2}},$ we have a reproducing kernel $\rho_{\frac{1}{2}}^z(z^\prime)=\frac{3}{8\pi}\left( \frac{z-\bar{z}^\prime}{2i}\right)^{-\frac{5}{2}}$ for the Hilbert space $\mathcal{H}_{\frac{1}{2}}$ and a reproducing kernel $\rho_{-\frac{1}{2}}^z(z^\prime)=\frac{1}{8\pi}\left( \frac{z-\bar{z}^\prime}{2i}\right)^{-\frac{3}{2}}$ for the Hilbert space  $\mathcal{H}_{-\frac{1}{2}}.$ Inverses of the unitary transforms are the maps $S^{\frac{1}{2}} : D_{\frac{1}{2}}\rightarrow \mathcal{H}_{\frac{1}{2}} $ and $S^{-\frac{1}{2}} : D_{-\frac{1}{2}}\rightarrow \mathcal{H}_{-\frac{1}{2}} $ respectively and are defined as 
$\left(S^{\frac{1}{2}}g\right)(z)  = 2^{\frac{5}{4}}(z+i)^{-\frac{5}{2}} g\left(\frac{z-i}{z+i}\right)$ and $ \left(S^{-\frac{1}{2}}g\right)(z)  = 2^{\frac{3}{4}}(z+i)^{-\frac{3}{2}} g\left(\frac{z-i}{z+i}\right)$ also the complete orthogonal bases for $\mathcal{H}_{\frac{1}{2}}$ and $\mathcal{H}_{-\frac{1}{2}}$ will be \[\ell^{\frac{1}{2}}_n(z) = 4\sqrt{\frac{\Gamma(n+\frac{5}{2})}{\pi^{\frac{3}{2}}n!}}\left(\frac{z-i}{z+i}\right)^n(z+i)^{-\frac{5}{2}}\] and \[\ell^{-\frac{1}{2}}_n(z) = \sqrt{\frac{2\Gamma(n+\frac{3}{2})}{\pi^{\frac{3}{2}}n!}}\left(\frac{z-i}{z+i}\right)^n(z+i)^{-\frac{3}{2}}\]
A large class of operators on $ \mathcal{H}_{\frac{1}{2}}$ can be expressed in terms of kernels, i.e., \[(Af)(z)  = \int h(z,\bar{z}^\prime)f(z^\prime)d\mu_{\frac{1}{2}}(z^\prime)\]
\section{Isometry between $\mathcal{S}(\mathbb{R})\subseteq L^2(\mathbb{R})$ and $\mathcal{H}_{\frac{1}{2}}\oplus\mathcal{H}_{-\frac{1}{2}}$}\large
Denote $\mathcal{S}_o(\mathbb{R}), \,\mathcal{S}_e(\mathbb{R})$ the collection of square integrable odd and even Schwartz class functions respectively. Choose 
$\varphi_o(p)=pe^{-p^2/2}\in \mathcal{S}_o(\mathbb{R})$  and $\varphi_e(p)=p^2e^{-p^2/2}\in \mathcal{S}_e(\mathbb{R})$ which satisfy the admissibility
condition. For any $\psi\in \mathcal{S}_o(\mathbb{R})$ assign \[\psi\mapsto \frac{1}{\sqrt{C_{\varphi_o}}}\langle\psi,U_1(a,b)\varphi_o\rangle=\frac{1}{2\pi^{3/4}}\langle\psi,U_1(a,b)\varphi_o\rangle\] and 
for $\psi\in \mathcal{S}_e(\mathbb{R})$ assign \[\psi\mapsto \frac{1}{\sqrt{C_{\varphi^e}}}\langle\psi,U_1(a,b)\varphi_e\rangle=\frac{1}{\sqrt{2}\pi^{3/4}}\langle\psi,U_1(a,b)\varphi_o\rangle\]
\begin{align*}
\frac{1}{a^{3/4}\sqrt{C_{\varphi_o}}}\langle\psi, U_1(a,b)\varphi_o\rangle&=\frac{1}{2\pi^{3/4}}\int_{-\infty}^{\infty}p\overline{e^{-i\bar{z}p^2/2}}\psi(p)dp\\
&= \frac{1}{2\pi^{3/4}}\int_{-\infty}^{\infty}pe^{izp^2/2}\psi(p)dp \\
\text{and}\qquad\frac{1}{a^{5/4}\sqrt{C_{\varphi_e}}}\langle\psi,U_1(a,b)\varphi_e\rangle&=\frac{1}{\sqrt{2}\pi^{3/4}}\int_{-\infty}^{\infty}p^2\overline{e^{-i\bar{z}p^2/2}}\psi(p)dp\\
&= \frac{1}{\sqrt{2}\pi^{3/4}}\int_{-\infty}^{\infty}p^2e^{izp^2/2}\psi(p)dp 
\end{align*}
\noindent with $z=b+ia.$ This suggests two transforms $A_o:\mathcal{S}_o(\mathbb{R})\rightarrow \mathcal{H}_{-\frac{1}{2}}, \,A_e:\mathcal{S}_e(\mathbb{R})\rightarrow \mathcal{H}_{\frac{1}{2}}$
\begin{align}\label{transformsoddeven}
(A_o\psi)(z)&=\frac{1}{2\pi^{3/4}}\int_{-\infty}^{\infty}pe^{izp^2/2}\psi(p)dp=\langle\psi, \varphi_o^z\rangle\\
(A_e\psi)(z)&=\frac{1}{\sqrt{2}\pi^{3/4}}\int_{-\infty}^{\infty}p^2e^{izp^2/2}\psi(p)dp=\langle\psi, \varphi_e^z\rangle
 \end{align}
where $\varphi_o^z=\frac{1}{a^{3/4}\sqrt{C_{\varphi_o}}}U_1(a,b)\varphi_o$ and 
$\varphi_e^z=\frac{1}{a^{5/4}\sqrt{C_{\varphi_e}}}U_1(a,b)\varphi_e$ 
\begin{remark}
Every function in $\mathcal{S}(\mathbb{R})$ can be decomposed into even and odd parts $\psi=\psi_e+\psi_o$ with $\psi_e\in \mathcal{S}_e(\mathbb{R})$ and 
$\psi_o\in \mathcal{S}_o(\mathbb{R}),$ so we have a map $A:\mathcal{S}(\mathbb{R})\rightarrow \mathcal{H}_{\frac{1}{2}}\oplus\mathcal{H}_{-\frac{1}{2}}$
\end{remark}
\[\begin{pmatrix}
   \psi_e\\
   \psi_o
  \end{pmatrix}\xrightarrow{A} \begin{pmatrix}A_e\psi_e\\
                                      A_o\psi_o\end{pmatrix}\]
\begin{theorem}\label{aoisometry}
 $A_o:\mathcal{S}_o(\mathbb{R})\rightarrow H_{-1/2}$ is a unitary transform and the inverse transform is given by 
\[(A_o^{-1}f)(p)=\frac{1}{2\pi^{3/4}}\lim_{\sigma,\gamma\rightarrow\infty} \int_{R}pe^{-i\bar{z}p^2/2}f(z)a^{-1/2}dadb\]
where $R=\{z\in\mathbb{C}:|Re(z)|<\sigma \hspace{0.1cm}\text{and}\hspace{0.1cm} \frac{1}{\gamma}<Im(z)<\gamma\}$
\end{theorem}
\begin{proof}
 $\langle A_o\varphi,A_o\psi\rangle_{\frac{-1}{2}}$
\begin{align*}
&=\frac{1}{4\pi^{3/2}}\int_{-\infty}^{\infty}\int_{0}^{\infty}\int_{-\infty}^{\infty}pe^{-i\bar{z}p^2/2}\overline{\psi(p)}dp
\int_{-\infty}^{\infty}p^{\prime}e^{izp^{\prime 2}/2}\varphi(p^\prime)dp^\prime a^{-1/2}dadb\\
&=\frac{1}{4\pi^{3/2}}\iiint_{-\infty}^{\infty}\int_{0}^{\infty}pp^{\prime}e^{-i\bar{z}p^2/2}
e^{izp^{\prime 2}/2}\overline{\psi(p)}\varphi(p^\prime)a^{-1/2}dpdp^\prime dadb\\
&=\frac{1}{4\pi^{3/2}}\iiint_{-\infty}^{\infty}\int_{0}^{\infty}pp^{\prime}e^{i\frac{b}{2}(p^{\prime2}-p^2)}
e^{-\frac{a}{2}(p^{\prime2}+p^2)}\overline{\psi(p)}\varphi(p^\prime)a^{-1/2}dpdp^\prime dadb\\
&=\iint_{-\infty}^{\infty}\int_{0}^{\infty}pp^{\prime}
e^{-\frac{a}{2}(p^{\prime2}+p^2)}\overline{\psi(p)}\varphi(p^\prime)\delta(p^2-p^{\prime 2})a^{-1/2}dpdp^\prime da\\
&=\iint_{-\infty}^{\infty}\int_{0}^{\infty}pp^{\prime}\left[\frac{\delta(p-p^\prime)}{|p|}+\frac{\delta(p+p^\prime)}{|p|}\right]
e^{-\frac{a}{2}(p^{\prime2}+p^2)}\overline{\psi(p)}\varphi(p^\prime)\frac{dpdp^\prime da}{\sqrt{a}}\\
&=\int_{-\infty}^{\infty}\int_{0}^{\infty}pe^{-ap^2}\overline{\psi(p)}\varphi(p)a^{-1/2}dpda\\
&=\int_{-\infty}^{\infty}\overline{\psi(p)}\varphi(p)dp=\langle\varphi,\psi \rangle_{L^2(\mathbb{R})}\\
\end{align*}
\begin{align*}
\parallel A_o^{-1}f(p)\parallel^2&=\langle A_o^{-1}f(p),A_o^{-1}f(p) \rangle\\
&=\frac{1}{4\pi^{3/2}}\int \int p^2e^{i(z^{\prime}-\bar{z})p^2/2}f(z)\overline{f(z^{\prime})}d\mu_{-\frac{1}{2}}(z)d\mu_{-\frac{1}{2}}(z^{\prime})\\
\int_{-\infty}^{\infty}\parallel A_o^{-1}&f(p)\parallel^2 dp\\
&=\frac{1}{4\pi^{3/2}}\int \int \int_{-\infty}^{\infty}p^2e^{-p^2(\frac{z^{\prime}-\bar{z}}{2i})}dpf(z)\overline{f(z^{\prime})}d\mu_{-\frac{1}{2}}(z)
d\mu_{-\frac{1}{2}}(z^{\prime})\\
&=\frac{1}{4\pi^{3/2}}\int \int \int_{0}^{\infty}u^{1/2}e^{-u(\frac{z^{\prime}-\bar{z}}{2i})}duf(z)\overline{f(z^{\prime})}d\mu_{-\frac{1}{2}}(z)
d\mu_{-\frac{1}{2}}(z^{\prime})
\end{align*}
\begin{align*}
&=\iint \frac{1}{8\pi}\left(\frac{z^{\prime}-\bar{z}}{2i}\right)^{-3/2}f(z)\overline{f(z^{\prime})}d\mu_{-\frac{1}{2}}(z)
d\mu_{-\frac{1}{2}}(z^{\prime})\\
&=\iint \rho_{-1/2}\left(z^{\prime}-\bar{z}\right)\overline{f(z^{\prime})}f(z)d\mu_{-\frac{1}{2}}(z)
d\mu_{-\frac{1}{2}}(z^{\prime})\\
&=\int \overline{f(z)}f(z)d\mu_{-\frac{1}{2}}(z)=\int |f(z)|^2 d\mu_{-\frac{1}{2}}(z)<+\infty\qedhere
\end{align*}
\end{proof}
\begin{theorem}\label{aeisometry}
 $A_e:\mathcal{S}_e(\mathbb{R})\rightarrow H_{1/2}$ is a unitary transform and the inverse transform is given by 
\[(A_e^{-1}f)(p)=\frac{1}{\sqrt{2}\pi^{3/4}}\lim_{\sigma,\gamma\rightarrow\infty} \int_{R}p^2e^{-i\bar{z}p^2/2}f(z)a^{1/2}dadb\]
where $R=\{z\in\mathbb{C}:|Re(z)|<\sigma \hspace{0.1cm}\text{and}\hspace{0.1cm} \frac{1}{\gamma}<Im(z)<\gamma\}$
\end{theorem}
\begin{proof}
 $\langle A_e\varphi,A_e\psi \rangle_{1/2}$
\begin{align*}
&=\frac{1}{2\pi^{3/2}}\iint_{-\infty}^{\infty}\int_{0}^{\infty}p^2e^{-i\bar{z}p^2/2}\overline{\psi(p)}dp
\int_{-\infty}^{\infty}p^{\prime 2}e^{izp^{\prime 2}/2}\varphi(p^\prime)dp^\prime \sqrt{a}dadb\\
&=\frac{1}{2\pi^{3/2}}\iiint_{-\infty}^{\infty}\int_{0}^{\infty}p^2p^{\prime 2}e^{-i\bar{z}p^2/2}
e^{izp^{\prime 2}/2}\overline{\psi(p)}\varphi(p^\prime)\sqrt{a}dpdp^\prime dadb\\
&=\frac{1}{2\pi^{3/2}}\iiint_{-\infty}^{\infty}\int_{0}^{\infty}p^2p^{\prime2}e^{i\frac{b}{2}(p^{\prime2}-p^2)}
e^{-\frac{a}{2}(p^{\prime2}+p^2)}\overline{\psi(p)}\varphi(p^\prime)\sqrt{a}dpdp^\prime dadb\\
&=\frac{1}{\sqrt{\pi}}\iint_{-\infty}^{\infty}\int_{0}^{\infty}p^2p^{\prime2}
e^{-\frac{a}{2}(p^{\prime2}+p^2)}\overline{\psi(p)}\varphi(p^\prime)\delta(p^2-p^{\prime 2})\sqrt{a}dpdp^\prime da\\
&=\iint_{-\infty}^{\infty}\int_{0}^{\infty}p^2p^{\prime2}\left[\frac{\delta(p-p^\prime)}{|p|}+\frac{\delta(p+p^\prime)}{|p|}\right]
e^{-\frac{a}{2}(p^{\prime2}+p^2)}\overline{\psi(p)}\varphi(p^\prime)\sqrt{a}dpdp^\prime da\\
&=\int_{-\infty}^{\infty}\int_{0}^{\infty}p^3e^{-ap^2}\overline{\psi(p)}\varphi(p)\sqrt{a}dpda\\
&=\int_{-\infty}^{\infty}\overline{\psi(p)}\varphi(p)dp=\langle \varphi,\psi \rangle_{L^2(\mathbb{R})}
\end{align*}
\begin{align*}
\parallel A_e^{-1}f(p)\parallel^2&=\langle A_e^{-1}f(p),A_e^{-1}f(p) \rangle\\
&=\frac{1}{2\pi^{3/2}}\iint p^4e^{i(z^{\prime}-\bar{z})p^2/2}f(z)\overline{f(z^{\prime})}d\mu_{\frac{1}{2}}(z)
d\mu_{\frac{1}{2}}(z^{\prime})
\end{align*}
\begin{align*}
\int_{-\infty}^{\infty}\parallel& A_e^{-1}f(p)\parallel^2 dp\\
&=\frac{1}{2\pi^{3/2}}\int \int \int_{-\infty}^{\infty}p^4e^{-p^2(\frac{z^{\prime}-\bar{z}}{2i})}dpf(z)\overline{f(z^{\prime})}d\mu_{\frac{1}{2}}(z)
d\mu_{\frac{1}{2}}(z^{\prime})\\
&=\frac{1}{2\pi^{3/2}}\int \int \int_{0}^{\infty}u^{3/2}e^{-u(\frac{z^{\prime}-\bar{z}}{2i})}duf(z)\overline{f(z^{\prime})}d\mu_{\frac{1}{2}}(z)
d\mu_{\frac{1}{2}}(z^{\prime})\\
&=\int \int \frac{3}{8\pi}\left(\frac{z^{\prime}-\bar{z}}{2i}\right)^{-5/2}f(z)\overline{f(z^{\prime})}d\mu_{\frac{1}{2}}(z)
d\mu_{\frac{1}{2}}(z^{\prime})\\
&=\int \int \rho_{1/2}\left(z^{\prime}-\bar{z}\right)\overline{f(z^{\prime})}f(z)d\mu_{\frac{1}{2}}(z)
d\mu_{\frac{1}{2}}(z^{\prime})\\
&=\int \overline{f(z)}f(z)d\mu_{\frac{1}{2}}(z)=\int |f(z)|^2 d\mu_{\frac{1}{2}}(z)<+\infty\qedhere
\end{align*}
\end{proof}
\section{Unitary Equivalent Operators on $\mathcal{H}_\frac{1}{2}\oplus\mathcal{H}_{-\frac{1}{2}}$}\large
An operator $H$ on $\mathcal{H}_\alpha,$ where $\alpha\in\{\frac{1}{2},-\frac{1}{2}\}$ can be expressed by a kernel, a function of two complex variables analytic
in the first and antianalytic in the second $(Hf)(z)=\int h(z,\bar{z}^\prime)f(z^\prime)d\mu_{\alpha}(z^\prime)$ where
the kernel $h(z,\bar{z}^\prime)=\langle \rho_\alpha^z,H\rho_\alpha^{z^\prime}\rangle.$ The kernel of the adjoint operator $H^*$ is $\overline{h(z^\prime,\bar{z})}=\langle H\rho_\alpha^z,\rho_\alpha^{z^\prime}\rangle.$
\begin{definition}
An operator $\widetilde{H}$ in $\mathcal{H}_{\frac{1}{2}}\oplus\mathcal{H}_{-\frac{1}{2}}$ can be defined which is unitary equivalent to an operator $H$ in $L^2(\mathbb{R})$ if $H$ does not change the parity of the function
\[\widetilde{H}= AHA^{-1}=\begin{pmatrix}
                        A_eHA_e^{-1}&0\\
                         0&A_oHA_o^{-1}
                       \end{pmatrix}
\]
and if $H$ change the parity of a function, then
\[ \widetilde{H}= AHA^{-1}=\begin{pmatrix}
                        0&A_eHA_o^{-1}\\
                         A_oHA_e^{-1}&0
                       \end{pmatrix}
\]
\end{definition}
\begin{proposition}\label{equivalentqp}
\begin{align*}
\widetilde{QP}= AQPA^{-1}&=\begin{pmatrix}
                        A_eQPA_e^{-1}&0\\
                         0&A_oQPA_o^{-1}
                       \end{pmatrix}\\&=-i\begin{pmatrix}
                       2z\frac{\partial}{\partial z}+3&0\\
                         0&2\frac{\partial}{\partial z}+2
                       \end{pmatrix} 
\end{align*}
\end{proposition}
\begin{proof}
For convenience we take $A_e\psi_e=f_e$ and $A_o\psi_o=f_o$
\begin{align*}
(A_eQPA_e^{-1})f_e(z)&=\langle QP\psi_e,\varphi_e^z\rangle=-\langle\psi_e,PQ\varphi_e^z\rangle\\
&=-i\langle\psi_e,\left(2\bar{z}\frac{\partial}{\partial \bar{z}}+3\right) \varphi_e^z\rangle=-i
\left(2z\frac{\partial}{\partial z}+3\right)\langle\psi_e,\varphi_e^z\rangle\\
(A_oQPA_o^{-1})f_o(z)&=\langle QP\psi_o,\varphi_o^z\rangle=-\langle\psi_o,PQ\varphi_o^z\rangle\\
&=-i\langle\psi_o, p\frac{\partial }{\partial p}\varphi_o^z+\varphi_o^z\rangle=-i\langle\psi_o,\left(2\bar{z}\frac{\partial}{\partial\bar{z}}+2\right)\varphi_o^z\rangle\\
&=-2i\left(z\frac{\partial}{\partial z}+1\right)\langle\psi_o,\varphi_o^z\rangle\qedhere
\end{align*}
\end{proof}
\begin{proposition} Multiplication operator $\zeta f(z)=zf(z)$ in $\mathcal{H}_{\frac{1}{2}}\oplus\mathcal{H}_{-\frac{1}{2}}$
\begin{equation}
\zeta A_e \approx A_e\left(\frac{Q^{-1}P+PQ^{-1}}{2}+\frac{3iQ^{-2}}{2}\right) 
\end{equation}
and 
\begin{equation}
\zeta A_o \approx A_o\left(\frac{Q^{-1}P+PQ^{-1}}{2}+\frac{iQ^{-2}}{2}\right) 
\end{equation}
\begin{proof}
 \begin{align*}
PQ^{-1}\varphi_o^z(p)&=-i\frac{\partial}{\partial p}\left(\frac{1}{p}\varphi_o^z(p)\right)\\&=\frac{-i}{p}\frac{\partial}{\partial p}\varphi_o^z(p)+\frac{i}{p^2}\varphi_o^z(p)\\
Q^{-1}P\varphi_o^z(p)&=\frac{-i}{p}\frac{\partial}{\partial p}\varphi_o^z(p)\\
\left[\frac{Q^{-1}P+PQ^{-1}}{2}\right]\varphi_o^z(p)&=\frac{-i}{p}\frac{\partial}{\partial p}\varphi_o^z(p)+\frac{i}{2p^2}\varphi_o^z(p)\\
\left[\frac{Q^{-1}P+PQ^{-1}}{2}-\frac{i}{2}Q^{-2}\right]\varphi_o^z(p)&=\frac{-i}{p}\frac{\partial}{\partial p}\varphi_o^z(p)
\end{align*}
If $\varphi_o^z(p)=\frac{1}{2\pi^{3/4}}p e^{-i\bar{z}p^2/2}$ then, 
\begin{align*}
\frac{i}{p}\frac{\partial}{\partial p}\varphi_o^z(p)&=\frac{i}{2p\pi^{3/4}}\frac{\partial}{\partial p}\left(pe^{-i\bar{z}p^2/2}\right)
=\frac{1}{2\pi^{3/4}}\left(\bar{z}p+\frac{i}{p}\right)e^{-i\bar{z}p^2/2}\\&=\left(\bar{z}+iQ^{-2}\right)\varphi_o^z(p)\\
\bar{z}\varphi_o^z(p)&=-\left[\frac{Q^{-1}P+PQ^{-1}}{2}+\frac{i}{2}Q^{-2}\right]\varphi_o^z(p)\\
z\overline{\varphi_o^z(p)}&=\overline{-\left[\frac{Q^{-1}P+PQ^{-1}}{2}+\frac{i}{2}Q^{-2}\right]\varphi_o^z(p)}
\end{align*}
\begin{align*}
(\zeta f)(z)&=zf(z)=\langle\psi_e+\psi_o,\bar{z}\varphi_e^z+\bar{z}\varphi_o^z\rangle=\langle\psi_e,\bar{z}\varphi_e^z\rangle+\langle\psi_o,\bar{z}\varphi_o^z\rangle
\end{align*}
\small{\begin{align*}
\zeta\begin{pmatrix}
 f_e\\
f_o
\end{pmatrix}=\begin{pmatrix}
                        A_e\frac{Q^{-1}P+PQ^{-1}}{2}+\frac{3iQ^{-2}}{2}A_e^{-1}&0\\
                         0&A_o\frac{Q^{-1}P+PQ^{-1}}{2}+\frac{iQ^{-2}}{2}A_o^{-1}
                       \end{pmatrix}\begin{pmatrix}
 f_e\\
f_o
\end{pmatrix} \end{align*}}
which gives the unitary equivalent operator correspoding to the multiplication operator in $\mathcal{H}_{\frac{1}{2}}\oplus\mathcal{H}_{-\frac{1}{2}}.$
\end{proof}
\end{proposition}
\bibliographystyle{siam}{}
\nocite{FAU,BVL}
\bibliography{myrefs}
\end{document}